\documentclass{llncs}
\newif\iflong
\newif\ifshort

\iflong
\else
\shorttrue
\fi

\usepackage[numbers,sort&compress]{natbib} 
\usepackage[breaklinks]{hyperref}
\usepackage{mathtools}
\usepackage{paralist}
\usepackage{microtype}
\usepackage[capitalize]{cleveref}
\usepackage{graphicx}
\usepackage{tikz}
\usepackage{booktabs}
\usepackage{bbm}
\usepackage{todonotes}
\usepackage{amsmath,amssymb}
\usepackage{multirow}
\usepackage{algorithm,algorithmic}
\usepackage{mdwlist}
\usepackage{comment}
\usepackage{tabulary}
\usepackage{tabularx}
\usepackage{graphics,color}

\newenvironment{mathenum}{%
  \begin{enumerate}[(i)]%
  }{%
  \end{enumerate}%
}

\spnewtheorem{rrule}[rrulec]{Rule}{\bfseries}{\itshape}
\crefname{rrule}{Rule}{Rules}
\crefname{table}{Table}{Tables}

\spnewtheorem{brule}[brulec]{Branching~Rule}{\bfseries}{\itshape}
\crefname{brule}{Branching~Rule}{Branching~Rules}
\crefname{property}{Property}{Properties}

\newcommand{\seePT}{\mathrm{See}(P,T)}
\newcommand{\friendP}{\mathrm{f}_P}
\newcommand{\friendPG}{\mathrm{f}_{{P_T}}}

\newcommand{\eqvP}{\sim_P}
\newcommand{\neqvP}{\not\sim_P}
\newcommand{\eqvQ}{\sim_Q}
\newcommand{\neqvQ}{\not\sim_Q}
\newcommand{\eqvPG}{\sim_{{P_T}}}

\newcommand{\Tflat}{{\mathcal{M}(T)}}
\newcommand{\up}[1]{\{#1\}}

\newcommand{\lii}[1]{{\small \textit{\phantom{1}#1}} \>}

\newenvironment{noSpaceTabbing}
  {\setlength{\topsep}{0pt}%
   \setlength{\partopsep}{0pt}%
   \tabbing}
  {\endtabbing}

\DeclareMathOperator{\poly}{poly}

\author{Laurent Bulteau\inst{1}, Guillaume Fertin\inst{1}, Christian
  Komusiewicz\inst{1}\thanks{Post-doc funded by a R\'egion Pays de la
    Loire grant} \\and Irena Rusu\inst{1}}

\institute{
  Universit\'e de Nantes, LINA - UMR CNRS 6241, France.
  \email{\{Laurent.Bulteau,Guillaume.Fertin,Christian.Komusiewicz,Irena.Rusu\}@univ-nantes.fr}
}

\title{A Fixed-Parameter Algorithm for Minimum Common String Partition with Few Duplications}
\begin{document}

\maketitle

\begin{abstract}
  Motivated by the study of genome rearrangements, the NP-hard \textsc{Minimum Common String Partition} problems asks,
  given two strings, to split both strings into an identical set of
  blocks. We consider an extension of this problem to unbalanced
  strings, so that some elements may not be covered by any block.  We
  present an efficient fixed-parameter algorithm for the parameters
  number~$k$ of blocks and maximum occurrence~$d$ of a letter in
  either string. We then evaluate this algorithm on bacteria genomes and
  synthetic data.
\end{abstract}

\section{Introduction} 
Comparative genomics has various applications, one of which is understanding the evolution of genomes under
the assumption that gene content and gene order conservation are
closely related to gene function \cite{overbeek1999use}. To this end,
a fundamental task is to define and compute the true evolutionary
distance between two given genomes \cite{SMEM08}. This is done by the correct identification of orthologs and 
paralogs and by the correct identification of the evolutionary events resulting into changes
in gene content and gene order. The first of these objectives is handled by several homology-based approaches~\cite{tatusov2001cog,remm2001automatic}; more evolved programs handle both objectives
\cite{CZF+05,FCV+07,SZJ10}. The second 
objective gave birth to a  large number of important distances between genomes represented either
as strings or as permutations. Such distances either exploit the similarity between genomes in
terms of gene content and order, or count specific genome
rearrangements needed to transform one genome into another
(see~\cite{FLR+09} for an extensive survey). 

In this work, both objectives above are followed \emph{via} a distance between genomes represented as
strings, which was defined independently by \citet{CZF+05} (for ortholog/paralog identification) and \citet{SMEM08} (for 
evolutionary events defining an evolutionary distance).  Informally, given two strings~$S_1$ and~$S_2$ representing two genomes, the
operation to realize is cutting~$S_1$ into non-overlapping substrings and reordering a subset of these 
substrings such that the concatenation of the reordered substrings is as close as possible to~$S_2$. 
The ortholog/paralog identification between~$S_1$ and~$S_2$ is then directly given by the substrings of
~$S_1$ used to approximately recompose~$S_2$, whereas the evolutionary distance is given by the minimum 
number of substrings needed to obtain such a reconstruction.

The above transformation between the two genomes is formalized by the notion of
common string partition (CSP). Let~$S_1$ and~$S_2$ be two strings on an alphabet
$\Sigma$. A partition~$P$ of~$S_1$ and~$S_2$ into \emph{blocks}~$x_1 x_2 \cdots x_p$ and~$y_1 y_2 \cdots y_q$
is a \emph{common string partition} if there is a bijective
function~$M$ from $D(M)\subseteq \{x_i\,|\, 1\le i\le p\}$ to
$I(M)\subseteq \{y_j\,|\, 1\le j\le q\}$ such that (1)~for
each~$x_i\in D(M)$, $x_i$ is the same string as~$M(x_i)$, and
(2)~there is no letter $a\in\Sigma$ that is 
simultaneously present in some block $x_j\not\in D(M)$ and in some
block $y_l\not\in I(M)$ (see \cref{fig:CSP} for an 
example). The \emph{size} of the common string partition $P$ is the cardinality~$k$ of $D(M)$.
We study the problem of finding a minimum-size CSP:
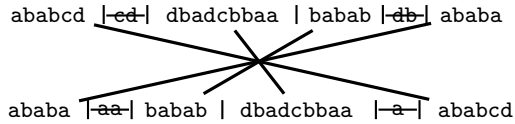
\begin{figure}[t]\centering
  \begin{tikzpicture}[>=stealth,x=33pt,y=18pt]
    \tikzstyle{inv}=[inner sep=0pt,outer sep=0pt,minimum size=0mm]
    \tikzstyle{vert}=[circle,draw,thick, inner sep=0pt,minimum
    size=2mm] \tikzstyle{subtree}=[regular polygon, regular polygon
    sides=3,draw,thick, inner sep=0pt,minimum size=15mm]
    \node (sx1) at (0.4,2) {\texttt{ababcd}}; \node (sx2) at (2.4,2)
    {\texttt{dbadcbbaa}}; \node (sx3) at (3.75,2) {\texttt{babab}};
    \node (sx4) at (5.25,2) {\texttt{ababa}};
  
    \node (sy1) at (5.3,0) {\texttt{ababcd}}; \node (sy2) at (3.25,0)
    {\texttt{dbadcbbaa}}; \node (sy3) at (1.87,0) {\texttt{babab}};
    \node (sy4) at (0.3,0) {\texttt{ababa}};

		\draw [-,thick] (1.04,2.2) -- (1.04,1.8);
    \draw [-,thick] (1.54,2.2) -- (1.54,1.8); 
		\draw [-,thick] (3.25,2.2) -- (3.25,1.8); 
		\draw [-,thick] (4.24,2.2) --  (4.24,1.8);
    \draw [-,thick] (4.74,2.2) --  (4.74,1.8);

    \draw [-,thick] (0.86,0.2) -- (0.86,-0.2);
		\draw [-,thick] (1.36,0.2) -- (1.36,-0.2); 
		\draw [-,thick] (2.4,0.2) -- (2.4,-0.2); 
		\draw [-,thick] (4.16,0.2) --    (4.16,-0.2);
     \draw [-,thick] (4.66,0.2) --    (4.66,-0.2);
		
			\node (del) at (1.29,2) {\texttt{cd}};
			\node (del) at (4.49,2) {\texttt{db}};
			\node (del) at (1.11,0) {\texttt{aa}};
		  \node (del) at (4.41,0) {\texttt{a}};

		\draw [-,thick] (1.07,1.98) --    (1.51,1.98);
     \draw [-,thick] (4.27,1.98) --    (4.71,1.98);
		\draw [-,thick] (4.19,0) --    (4.63,0);
     \draw [-,thick] (0.89,0) --    (1.33,0);

    \draw [-,very thick] (sx1) -- (sy1); 
		\draw [-,very thick] (sx2) --
    (sy2); \draw [-,very thick] (sx3) -- (sy3); \draw [-,very thick]
    (sx4) -- (sy4); 
  \end{tikzpicture}
	\caption{\label{fig:CSP}A common string partition of size 4. 
	Copies of \texttt{a}, \texttt{b}, \texttt{c} and \texttt{d} that could not be matched
	are deleted.}
	\end{figure}
\begin{quote}
  \textsc{Minimum Common String Partition (MCSP)}\\
  {\bf Input:} Two strings~$S_1$ and~$S_2$ on an alphabet $\Sigma$, and an integer~$k$.\\
  {\bf Question:} Is there a common string partition (CSP) of~$S_1$ and~$S_2$ of size at most~$k$?
\end{quote}

The definition of a CSP given above is actually a generalization to
arbitrary (or {\em unbalanced}) strings of the definition given in
\cite{CZF+05} for {\em balanced} strings, that is, when each letter
appears the same number of times in $S_1$ and $S_2$. Note also that in
this paper, the strings we consider are unsigned. Although this model
is less realistic from a genomic viewpoint, our study is a first step
towards improved algorithms for the MCSP problem in the most general
case, that is, for signed and unbalanced strings. 

\paragraph*{Related Work.}
MCSP was introduced by~\citet{CZF+05}, but close variants also exist
with different names, such as {\em block edit
  distance}~\cite{LT97} or {\em sequence
  cover}~\cite{SMEM08}. Most of the literature on MCSP actually
considers the restricted case where the input strings $S_1$ and $S_2$
are balanced. In that case, necessarily $D(M)$ (resp. $I(M)$) contains  
every block from $S_1$ (resp. $S_2$). Let Bal-MCSP denote this
restricted class of problems. Bal-MCSP has been shown to be NP-hard
and APX-hard even if $d=2$, where $d$ is the maximum number
of occurrences of any letter in either input string~\cite{GKZ05}.
Several approximation algorithms exist with ratios 1.1037 when $d=2$~\cite{GKZ05}, 4 when $d=3$~\cite{GKZ05}, and $4d$ in general~\cite{KW07}.
Concerning fixed-parameter tractability issues, \citet{Dam08}
initiated the study of Bal-MCSP in the context of parameterized algorithmics by showing
that it is fixed-parameter tractable with respect to the combined
parameter ``partition size~$k$ and repetition number~$r$''. More recently, \citet{JZZ+12} showed that Bal-MCSP can be solved
in~$O((d!)^k\cdot \poly(n))$ time. 

\paragraph{Our Results.}
Our main result in this paper is an improvement on the
latter result, showing that MCSP (and thus, Bal-MCSP) can be solved
in $O(d^{2k}\cdot kn)$ time, thus considerably improving the
running time from~\citet{JZZ+12}. Our result is also more general since
it is one of the rare known fixed-parameter algorithms that deals
with unbalanced strings.
Moreover, a(n approximate) solution to
MCSP is computed within the pipeline of MSOAR, MSOAR2.0 and MultiMSOAR
software~\cite{FCV+07,SZJ10,SPJ11} (all used to determine
orthology relations between genes), hence these programs could benefit from any algorithmic improvement
concerning MCSP~\cite{Jiang10}, such as the one presented here. Indeed, our
algorithm actually runs in $d^{2k'}\cdot kn$,
where $k'$ is the number of blocks of $D(M)$ that contain {\em no letter 
appearing only once} in $S_1$ and~$S_2$. Moreover, we present reduction rules that
yield further speed-up, and finally test our algorithm on
genomic and synthetic data. 

\paragraph{Basic Notation.}
A \emph{marker} is an occurrence of a letter at a specific position in
a string. 
Formally, the marker at position~$i$ in a string~$S$ corresponds
to the pair $(S,i)$, which we denote by~$S[i]$.
Given a marker~$u$ we denote by 
$S(u)$ the string that contains~$u$.
For all~$i$, $1\le i<n$, the markers~$S[i]$ and~$S[i+1]$ are called
\emph{consecutive}. Let~$r(S[i]):=S[i+1]$,~$1\le i<n$, denote the
right neighbor of marker~$S[i]$, and let~$l(S[i]):=S[i-1]$,~$1<i\le n$
denote the left neighbor of marker~$S[i]$. An \emph{adjacency} is a
pair of consecutive markers. For two markers~$u$ and~$v$ we
write~$u\equiv v$ if their letters are the same and~$u=v$ if the
markers are identical, that is, they are at the same position in the
same string.  An \emph{interval} is a set of consecutive markers, that
is, an interval is a set~$\{S[i], S[i+1], \ldots , S[j]\}$ for
some~$i\le j$. We write $[u,v]$ to denote the interval whose first
marker is~$u$ and whose last marker is~$v$. For two intervals~$s$
and~$t$, we write $s\equiv t$ if they represent the same string of
letters (if they have the same contents) and~$s=t$ if they are the
same interval, that is, they start and end at the same position in the
same string. Given two strings $S_1,S_2$,
a letter is \emph{abundant} in a string $S_i$ if it appears
with strictly more occurrences in $S_i$ than in the other string.
Otherwise, it is \emph{rare} in $S_i$.
A marker $u$ is \emph{abundant} 
if it corresponds to an abundant letter in $S(u)$, and \emph{rare}
otherwise.

\paragraph{Fundamental CSP-Related Definitions.}
We assume that~$S_1\neq S_2$, otherwise MCSP is trivially solved by reporting a CSP of size one.
A \emph{candidate match} is an unordered pair of markers $\up{u,v}$
such that $u\equiv v$ and $S(u)\neq S(v)$, that is, the markers have
the same letters and are from different input strings. Two candidate
matches $\up{x,y}$ and $\up{x',y'}$ where~$S(x)=S(x')$ and~$x$ is to
the left of~$x'$ are called \emph{parallel} if~$[x,x']\equiv
[y,y']$. Note that this implies that for the~$i$-th marker~$u$
in~$[x,x']$ and the~$i$-th marker~$v$ in~$[y,y']$ the pair~$\up{u,v}$
is also a candidate match and it is parallel to~$\up{x,y}$ and
to~$\up{x',y'}$. Informally, being parallel means that two 
candidate matches could potentially be in the same block of a CSP.

A \emph{CSP}~$P$ is a set  of pairwise disjoint candidate
matches containing all rare markers. 
If a marker does not appear in
any candidate match of~$P$ then it is necessarily abundant, 
and it is called \emph{deleted} in~$P$,
otherwise we
use $\friendP(u)$ to denote the unique marker $v$ such
that $\up{u,v}\in P$. 
The \emph{block relation} $\eqvP$ of a CSP is defined
as the (uniquely determined) equivalence relation such that each
equivalence class is a substring of~$S_1$ or~$S_2$ and $u\eqvP r(u)$
if and only if $u$ and $r(u)$ are not deleted, and 
$\up{u,\friendP(u)}$ and~$\up{r(u),\friendP(r(u))}$ are
parallel. Note that this implies that, for any two markers~$x$
and~$x'$ with $x\eqvP x'$ it holds that $\up{x,\friendP(x)}$ and
$\up{x',\friendP(x')}$ are parallel. 
The blocks are precisely the equivalence classes of
$\eqvP$ of non-deleted markers, that is, two markers $u$ and $v$ are in the same block iff
$u\eqvP v$.

Due to lack of space, some proofs are deferred to a full version of this work.

\section{An Improved Fixed-Parameter Algorithm}
\label{sec:search-tree}
We now describe our fixed-parameter algorithm. It is a branching algorithm that adds, one by one, candidate matches to a temporary solution. The
main idea is that these candidate matches belong to different blocks
of the CSP.
\subsection{CSPs, Samples and Witnesses}
\label{sec:search-tree-basics}
As stated above, the algorithm gradually extends a temporary solution
called sample. Formally, a \emph{sample} $T$ is a set of disjoint
candidate matches.  We use $\Tflat$ to denote the set of all markers
belonging to a candidate match in $T$ (thus, $|\Tflat|=2|T|$).
The algorithm tries to construct an optimal CSP by extending a
sample~$T$ that describes this CSP and is furthermore
non-redundant. That is, the sample  contains only candidate matches
that are in the CSP and at most one candidate match for each pair of
matched blocks. We call such samples witnesses.
\begin{definition}
  A sample $T=\{\up{x_1,y_1},\up{x_2,y_2}, \ldots, \up{x_m,y_m}\}$ is
  a \emph{witness} of a CSP $P$ if
  (1) $T\subseteq P$, that is, $y_i = \friendP(x_i)$ for each
    $i$, and (2) for all $x,y\in \Tflat$ with $x\neq y$ we have $x \neqvP y$.
\end{definition}
Given a witness~$T$ of some CSP $P$, a marker $u$ is 
\emph{seen by $T$} if $\exists x\in \Tflat$ such that $u\eqvP x$. We
use~$\seePT$ to denote the set of markers seen by $T$ in~$P$. Let
$u\in \seePT$ be a marker seen by~$T$ in~$P$, then we say that~$u$ is
\emph{colored black by~$P$ and~$T$} if $u=x$; $u$ is \emph{colored
  green by~$P$ and~$T$} if it is to the right of $x$; or $u$ is
\emph{colored red by~$P$ and~$T$} if it is to the left of $x$. Note
that the coloring is unique since for each marker~$u$ there is at most
one $x\in \Tflat$ such that $u\eqvP x$.

The algorithm finds a witness 
describing an optimal CSP. More precisely, the aim is to see all
rare markers eventually. A witness $T$ is \emph{complete} if it contains a
marker from every block of $P$. Equivalently, $T$ is complete if it
sees every rare marker.
We first show that if a rare marker is unseen by a witness~$T$
for some CSP~$P$, then another witness for~$P$ can be obtained by extending~$T$.
\begin{lemma}\label{lem:unseen}
  Let $u$ be a rare marker such that $u\notin\seePT$. Then there exists a
  candidate match~$\up{u,v}$ such that $T\cup\{\up{u,v}\}$ is a
  witness of $P$.
\end{lemma}
\begin{proof}
  Let $v=\friendP(u)$ ($u$ is rare, hence it is not deleted), 
	then $\up{u,v}$ is clearly a candidate match.
  Furthermore,~$T\cup\{\up{u,v}\}$ is a subset of $P$. It thus remains
  to show that~$T$ is non-redundant. Since $u\notin\seePT$, $u\neqvP
  x$ for all $x\in\Tflat$. Furthermore, this also implies~$v\neqvP y$
  for all $y\in\{\friendP (x) \mid x\in\Tflat\} = \Tflat$.  Thus
  $T\cup\{\up{u,v}\}$ is a witness of $P$. \qed
\end{proof}
The following lemma shows that when an optimal CSP contains parallel
candidate matches, then the markers that are in the same string are
also in the same blocks of the CSP. We will use this lemma to argue
that the algorithm only considers samples without parallel edges.
\begin{lemma} \label{lem:OptKeepParallels} If a CSP $P$ contains two
  parallel candidate matches $\up{x,y}$ and $\up{x',y'}$ such that
  $S(x)=S(x')$ and $x\neqvP x'$, then it is not optimal.
\end{lemma}
\begin{proof}
  
	Aiming at a contradiction, assume that $P$ is optimal. Moreover, 
	assume without loss of generality that $S(x)=S(x')=S_1$, and that 
	$\up{x,y}$ and $\up{x',y'}$ have been chosen so as 
 to minimize 
 the distance between $x$ and $x'$,
 while satisfying the conditions of the lemma. 
	Since the
  candidate matches $\up{x,y}$ and $\up{x',y'}$ are parallel, we
  have~$[x,x']\equiv [y,y']$. Let~$\ell$ denote the number of markers
  in~$[x,x']$, let~$x_i$ denote the~$i$-th marker in~$[x,x']$ and
  let~$y_i$ denote the~$i$-th marker in~$[y,y']$. Then, each
  $\up{x_i,y_i}$ is a candidate match, $\up{x_i,y_i}$ and
  $\up{x_j,y_j}$ are parallel for all $1\leq i,j\leq \ell$, and, by
  the minimality of the distance between~$x$ and~$x'$,
  $\up{x_i,y_i}\notin P$ for $1<i<\ell$. Moreover, for all $1<i<\ell$, 
	$x\neqvP x_i\neqvP x'$ and $y\neqvP y_i\neqvP y'$.
	 Create a CSP $Q$, starting with $Q:=P$.
	

 If one of $x_2$, $y_2$ is deleted (say $x_2$, note that they cannot 
 both be deleted since they cannot both be abundant), then let $u_2:=\friendP(y_2)$.
 The pair~$\up{u_2,y_2}$ is the left-most candidate match of its block in $P$. 
 Remove $\up{u_2,y_2}$ from~$Q$ and add $\up{x_2,y_2}$, extending the block containing~$\up{x,y}$. 
Then $Q$ is also an optimal CSP.

If none of $x_2$, $y_2$ are deleted, then they are the left-most markers of 
blocks ending in $x_p$ and $y_q$ respectively 
(assume without loss of generality that 
$p\leq q$). Note that $p,q< \ell$, since these blocks are strictly contained between $x$ and $x'$ ($y$ and $y'$). 
Write $u_i=\friendP(y_i)$ for all $2\leq i\leq q$, and $v_i=\friendP(x_i)$ for all $2\leq i\leq p$. For each $2\leq i\leq p$, remove 
$\{\up{x_i,v_i},\up{y_i,u_i}\}$ from $Q$ and add 
 $\{\up{x_i,y_i},\up{u_i,v_i}\}$. 
Then $Q$ has no more blocks than $P$ and is an optimal CSP. 
Indeed, $[x_2,x_p]$ is now merged to the
block containing $x$, and $[u_2,u_q]$ is now split 
in two blocks $[u_2,u_p]$ and $[u_{p+1},u_q]$. 

In both cases, $Q$ is an optimal CSP where $\up{x_2,y_2}$ has been added to the block containing $\up{x,y}$. 
If  $x_2\eqvQ x'$, then the block containing  $\up{x,y}$ and  $\up{x_2,y_2}$ 
is merged with $\up{x',y'}$, and $Q$ has one block less than~$P$. 
Otherwise, $x_2\neqvQ x'$, and $Q$ satisfies the conditions of the lemma for $\up{x_2,y_2}$ and $\up{x',y'}$ with a smaller distance between $x_2$ and $x'$ than between $x$ and $x'$. 
Both cases lead to a contradiction.
\qed
\end{proof}

\subsection{The Sample Graph}
\label{sec:search-tree-sample-graph}
We now describe a multigraph that is associated with the current
sample~$T$. We will use the structure in this graph to identify cases
to which the branching applies. First, we describe the construction of
this graph.

Let $T$ be a sample for
an input instance~$(S_1,S_2,k)$, and~let~$C$ denote the set of all
candidate matches between~$S_1$ and~$S_2$.  The \emph{sample
  graph}~$G_T:=\{V_T,E_T\}$ of~$T$ is the following edge-colored
multigraph. The vertex set~$V_T$ is the set of markers of $S_1$ and
$S_2$. The edge multiset~$E_T\subseteq C$ consists of the black
edges~$E_T^b$, the green edges~$E_T^g$, and the red edges
$E_T^r$. The edge sets are defined as follows. The black
edges are the pairs of the sample\iflong
\else
, that is, $E_T^b:=T$.
\fi
For the green and red edges, we use the following notation. For a
marker~$u\notin \Tflat$, let~$l_T(u)$ denote the rightmost vertex
from~$\Tflat$ that is in the same string as~$u$ and to the left
of~$u$. Similarly, let~$r_T(u)$ denote the leftmost vertex from~$\Tflat$
that is to the right of~$u$. Now, the green edge set
is \begin{align*}
  E_T^g:=\{\up{x,y}\in C\mid x,y\notin \Tflat & \wedge \up{l_T(x),l_T(y)}\in T\\  & \wedge \up{l_T(x),l_T(y)} \text{ is parallel to $\up{x,y}$} \}.
\end{align*} 
The red edge set is 
\begin{align*}
  E_T^r:=\{\up{x,y}\in C\mid x,y\notin \Tflat & \wedge \up{r_T(x),r_T(y)}\in T\\  & \wedge \up{r_T(x),r_T(y)} \text{ is parallel to $\up{x,y}$} \}.
\end{align*} 

\begin{figure}[t]
\centering
 \includegraphics[scale=0.9]{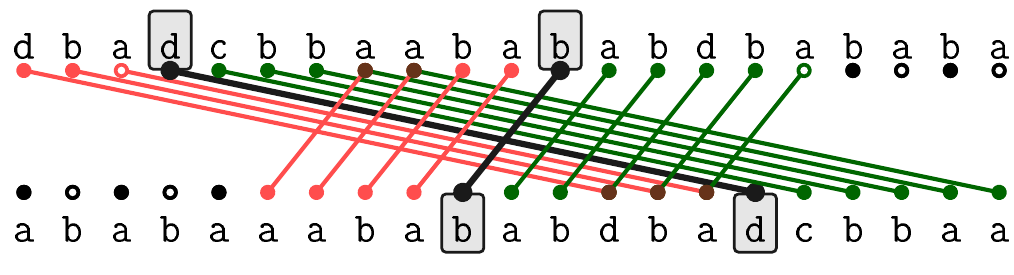}
\caption{\label{fig:sample-graph} Sample graph computed for two sequences, given a sample of two candidate matches (black edges), with green (dark gray) and red (light gray) edges. Note that {\tt a} and {\tt c} are rare in the top sequence, and {\tt b}, {\tt c} and {\tt d} are rare in the bottom sequence. Vertices satisfying the conditions of~\cref{brule:deg-zero,brule:odd-path} are marked with white dots (four are isolated rare vertices, two appear in a rare odd path).  }
\end{figure}

See \cref{fig:sample-graph} for an example. Clearly, $G_T$ is bipartite. From now on, we use the terms
``marker'' and ``vertex'' equivalently since there is a one-to-one
correspondence between them. Further, any definition applying to
candidate matches applies in a similar manner to edges.
The black-, green-, and red-degree of a vertex are the
number of black, green, and red edges incident with it. The degree of
a vertex is simply defined as the sum of the three colored degrees.
The sample graph has the following properties.
\begin{property}\label{prop:fillGaps}
  Let $\up{u,v}$ be a green (red) edge of $G_T$, then
  $\up{l(u),l(v)}$ ($\up{r(u),r(v)}$) is either a black or green (red)
  edge of~$G_T$.
\end{property}
\begin{proof}
  Consider the case that edge~$\up{u,v}$ is green. The property clearly
  holds if~$\up{l(u),l(v)}$ is black. Otherwise,~$\up{l(u),l(v)}$ also
  fulfills the conditions in the construction of~$E_T^g$:
  First,~$l(u)\neq l_T(u)$ and~$l(v)\neq l_T(v)$, thus they cannot
  belong to~$T$. Second,~$\up{l(u),l(v)}$ is to the left of~$\up{u,v}$
  and thus it is also parallel to~$\up{l_T(u),l_T(v)}$). \qed
\end{proof}
\begin{property}\label{prop:deg}
  Each vertex incident with a black edge has degree one. For each other vertex, green-degree and red-degree are at most one.
\end{property}
\begin{proof}
  First, let $\up{x,y}\in T$ be a black edge. By the definitions of~$E_T^g$
  and~$E_T^r$, neither~$x$ nor~$y$ is incident with a red or green
  edge. Since the sample~$T$ has only pairwise disjoint candidate
  matchings, there is no other black edge in~$T$ incident with
  either~$x$ or~$y$.

  Now, let~$e_1, e_2$ be two green edges incident with some
  vertex~$v$. Clearly,~$e_1$ and~$e_2$ fulfill the conditions in the
  definition of~$E_T^g$. Note that, by Property~\ref{prop:deg},~$l_T(v)$
  has degree one. Hence,~$e_1$ and~$e_2$ are parallel to the same
  edge. This implies $e_1=e_2$. The proof for red edges is
  symmetrical. \qed
\end{proof}

Property~\ref{prop:deg} implies that every vertex has degree at most two. Thus, each connected component is either a singleton, a path or a cycle.
\begin{property}\label{prop:parallelEdges}
Let $u$ and $u':=l(u)$ be two consecutive markers such that~$G_T$ contains the edges $\up{u,v}$ and $\up{u',v'}$. If both edges are green (both edges are red), then $\up{u,v}$ and $\up{u',v'}$ are parallel, that is, $v'=l(v)$.
\end{property}

\begin{proof}
  Assume that $\up{u,v}$ and $\up{u',v'}$ are
  green. By Property~\ref{prop:deg}, vertices incident with black edges
  have degree one. Hence,~$l_T(u)\neq u'$ and thus~$l_T(u) =
  l_T(u')$. Consequently, $\up{u,v}$ and $\up{u',v'}$ are parallel to
  the same edge~$\up{l_T(u),l_T(v)}$. Hence, they are also parallel to
  each other. The proof for red edges works analogously. \qed
\end{proof}

\subsection{Branching on Odd Connected Components}
\label{sec:search-tree-branch}
We now show some further properties that the sample graph~$G_T$ has
with respect to any CSP witnessed by the sample~$T$. We then exploit
these properties to devise branching rules that branch into~$O(d^2)$
cases.
Hence, consider an arbitrary CSP~$P$ witnessed by~$T$.
The following is a simple corollary
of \cref{lem:OptKeepParallels}, the construction of the sample graph, 
and the definition of witness.
\begin{lemma} \label{lem:doubleGR}
If $G_T$ contains two parallel black edges, then $P$ is not optimal.
\end{lemma}

The following lemma relates the colors that markers receive by the
CSP~$P$ to the edge colors in the sample graph. 
\begin{lemma} \label{lem:VEcolor} Let $u\in
  \seePT$ be a marker seen by~$T$. Then, there is at least one edge incident with $u$ in
  $G_T$. In particular, if vertex~$u$ is colored black/green/red,
  then~$\up{u,\friendP(u)}$ is a black/green/red edge in
  $G_T$.
\end{lemma}
\begin{corollary} \label{cor:degree0}
  If some vertex $u$ has degree 0 in $G_T$, then $u\notin \seePT$.
\end{corollary}
Combined with \cref{lem:unseen} this leads to the first
branching rule.
\begin{brule}\label{brule:deg-zero}
  If the sample graph~$G_T$ contains a rare degree-0 vertex~$u$, then for
  each vertex~$v\notin \Tflat$ such that~$S(u)\neq S(v)$ and~$u\equiv
  v$ branch into the case to add~$\up{u,v}$ to~$T$.
\end{brule}
The branching rule above deals with connected components that are
singletons. Next, we develop a branching rule for connected
components that are a certain type of path in the sample graph.

To this end, we distinguish the following types of paths. A \emph{black
  path} is a path containing exactly one black edge. An \emph{odd
  path} is a path with an odd number of vertices. An \emph{even path}
is a path with an even number of vertices. Note that
by \cref{prop:deg}, all black edges are contained exclusively
in black paths. Furthermore, also by \cref{prop:deg}, the colors of
each other path alternate between green and red. We thus call an even path
\emph{green} if it starts and ends with a green edge and \emph{red},
otherwise.
An odd path is \emph{abundant} if its first marker is abundant, \emph{rare} otherwise (this definition is not ambiguous since the markers 
at the ends of an even path correspond to the same letter in the same string, they are thus both abundant or both rare).
 
\begin{lemma} \label{lem:oddPath} Let~$T$ be a sample witnessing a
  CSP~$P$. If a connected component of the sample graph~$G_T$
  is a rare odd path $(u_1,v_1, \ldots, v_{\ell-1},u_\ell)$, then
  there is some~$u_i$, $1\leq i\leq \ell$, such that $u_i\notin \seePT$ and
  $u_i$ is not deleted in $P$.
\end{lemma}
\begin{brule}\label{brule:odd-path}
  If the sample graph~$G_T$ contains a connected component which is a
  rare odd path~$(u_1,v_1,u_2,v_2, \ldots, v_{\ell-1},u_\ell)$, then do the following for each vertex~$u_i$,~$1\le i\le \ell$: 
	for each vertex~$x\notin \Tflat$ such that~$S(u_i)\neq S(x)$
  and~$u_i\equiv x$ branch into the case to add~$\up{u_i,x}$ to~$T$.
\end{brule}

\subsection{Solving Instances without Rare Odd Paths or Singletons}\label{sec:construct-opt}
We now show how to find an optimal CSP in the remaining cases. As we will show, the edge set defined as follows gives such an optimal CSP. See~\cref{fig:graph-pt} for an example.

\begin{definition} \label{def:PT} Let~$G_T$ be a sample graph. The
  set~${P_T}$ is the edge set containing
  \begin{itemize}
  \item all black edges from~$G_T$, 
  \item each green edge that is in a green path, in an odd path or in a cycle, and
  \item each red edge that is in a red path.
  \end{itemize}
\end{definition}

\begin{lemma}\label{lem:pt-optimal}
Let~$T$ be a sample such that~$G_T$ does not contain isolated vertices, parallel black edges, or rare odd paths. 
Then,~${P_T}$ is a CSP for which~$T$ is a complete witness. 
\end{lemma}

\begin{figure}[t]
\centering
 \includegraphics[scale=0.9]{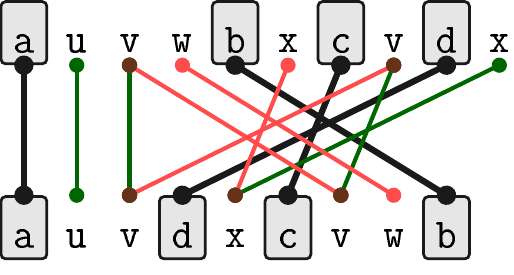}\hfill
 \includegraphics[scale=0.9]{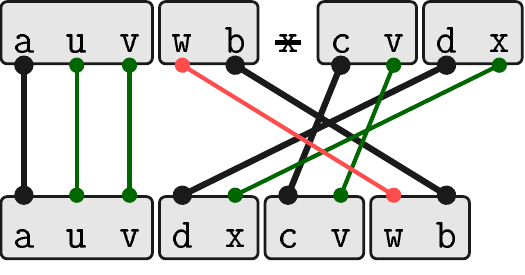}
\caption{\label{fig:graph-pt} Left: Sample graph $G_T$ with no isolated vertices, 
parallel black edges or rare odd paths. Right: CSP $P_T$ obtained from $G_T$ (\cref{def:PT,lem:pt-optimal}). 
Note that markers (with letter) {\tt u} form a green path, markers {\tt v} form a cycle, markers {\tt w} form a red path,
and markers {\tt x} form an abundant odd path.}
\end{figure}

\begin{algorithm}[t]
  \caption{{\small The fixed-parameter algorithm for parameter $(d,k)$.}}
  \label{algo:fpt-dk}
  \begin{noSpaceTabbing}
    \hspace*{0.7cm}\=\hspace*{0.5cm}\=\hspace*{0.5cm}\=\hspace*{0.6cm}\=\hspace*{4.3cm}\=\kill
    \texttt{MCSP}$(S_1,S_2,k,T)$\\
    \lii{1} \textbf{if} $|T|>k$ \textbf{abort branch}\\
    \lii{2} Compute the sample graph~$G_T$ \\
    \lii{3} \textbf{if}~$G_T$ contains parallel black edges \textbf{: abort branch}\\
    \lii{4} \textbf{else if}~$G_T$ contains an isolated vertex \textbf{:}\\
    \lii{5} \> \textbf{apply} \cref{brule:deg-zero}; in each case \textbf{call} \texttt{MCSP}$(S_1,S_2,k,T\cup \{\{u,v\}\})$ \\
    \lii{6} \textbf{else if}~$G_T$ contains a rare odd path \textbf{:}\\
    \lii{7} \> \textbf{apply} \cref{brule:odd-path}; in each case \textbf{call} \texttt{MCSP}$(S_1,S_2,k,T\cup \{\{u_i,x\}\})$\\
    \lii{8} \textbf{else} compute~$P_T$, \textbf{output~$P_T$}  
  \end{noSpaceTabbing}
\end{algorithm}

\begin{theorem}
  MCSP can be solved in~$O(d^{2k}\cdot kn)$ time.
\end{theorem}
\begin{proof} We use the algorithm \texttt{MCSP} outlined in Algorithm~\ref{algo:fpt-dk}.
  We first show the correctness of \texttt{MCSP}, then we bound the running time.
  Consider a yes-instance, and let~$P$ be an optimal CSP of
  size~$k$. We show that~\texttt{MCSP$(S_1,S_2,k,\emptyset)$} outputs at least
  one CSP of size~$k$ in this case. Since~$T=\emptyset$~$T$ in the first call,~$T$ is initially a witness of~$P$. Combining~\cref{lem:unseen} with~\cref{lem:VEcolor,lem:oddPath} shows that the algorithm creates in each
  application at least one branch such that the set~$T$ is a witness
  of~$P$ in this branch.
  Now, note that if a branch is aborted because~$|T|>k$, then the
  current set~$T$ either is redundant (and thus not a sample) or any
  CSP that it witnesses has size at least~$k$, thus it does not
  witness~$P$ in this case. Similarly, if the graph~$G_T$ contains
  parallel black edges, then the set~$T$ is either redundant, or any
  CSP that it witness is not optimal; thus it does not witness~$P$.
  Hence, the algorithm eventually reaches a situation in which~$T$ is
  a witness of~$P$ and~$G_T$ contains no isolated vertices and no odd
  paths. Then it constructs and outputs a
  set~$P_T$. By~\cref{lem:pt-optimal},~$P_T$ is a CSP. Furthermore, it
  has size~$|T|$ and thus it is at most as large as~$P$ which also has
  size at least~$|T|$ since~$T$ is a witness for~$P$.

  Now, assume that the instance is a no-instance, then the algorithm
  has empty output since all CSPs that are output have size at
  most~$k$ due to the condition in Line~1 of the algorithm.

  It remains to bound the running time. We first bound the size of the search
  tree. After the application of each branching rule, the set~$T$ has
  contains one additional candidate match, so the depth of the search
  tree is at most~$k$ because of the check in Line~1 of the
  algorithm. We now bound the number of new cases for each branching
  rule. First,~\cref{brule:deg-zero} branches into at most~$d$ cases.
  Second,~\cref{brule:odd-path} branches into at most~$d^2$ cases: All
  vertices of a path have the same letter since edges are candidate
  matches. Hence, there are at most~$d$ $u_i$'s. For each of them the algorithm
  creates at most~$d$ branches. Hence, the overall search tree size
  is~$O((d^2)^k)=O(d^{2k})$. The time spent in each search tree node
  can be seen as follows. The sample graph can be
  constructed in~$O(kn)$ time by adding for each of the~$O(k)$ black edges
  the red and green edges in linear time. This is done by moving
  to the left and the right until either the next parallel marker pair
  is not a candidate match or contains a black vertex (this can also be used to find parallel black edges). The sample
  graph has size~$O(n)$, hence isolated vertices and odd paths
  can also be found in~$O(n)$ time.  \qed
\end{proof}

\section{Parameter Improvement}\label{sec:param-k'}
In this section, we show that the parameter~$k$ denoting the number of
blocks in an optimal solution can be replaced by a potentially much
smaller parameter~$k':=$``number of blocks without unique
letters''. Herein, a letter is called \emph{unique} if it appears at
most once in~$S_1$ and at most once in~$S_2$. To deal with the blocks
that contain unique letters we devise a simple rule for simplifying
the instance. The algorithm makes use of a data reduction rule. A data reduction rule is \emph{correct} if the new instance is a yes-instance if and only if the old one is.  An instance is \emph{reduced} with respect to a data reduction rule, if an application of the rule does not change the instance.
\begin{rrule}\label{rule:reduce-unique}
  If~the input contains a pair of unique letters~$x$ and~$x'$,
  where~$x'$ is to the right of~$x$, such that the candidate
  matches~$\up{x,y}$ for~$x$ and~$\up{x',y'}$ for~$x'$ are parallel,
  then replace~$[x,x']$ by~$x$ and~$[y,y']$ by~$y$.
\end{rrule}
\begin{lemma}\label{lem:rule1}
  \cref{rule:reduce-unique} is correct.
\end{lemma}
After this simplification, the resulting instance has the property
that candidate matches between two different unique matches are in
different blocks. This implies that a set~$T$ containing all different
unique matches is a sample witnessing any optimal CSP. This leads to the following.
\begin{theorem}\label{thm:fpt-k'}
  MCSP can be solved in~$O(d^{2k'}\cdot kn)$~time where~$k'$
  denotes the number of blocks in~$S_1$ that contain no unique letter.
\end{theorem}

\section{Data Reduction Rules}\label{sec:data-red}

In addition to the improvements described in previous sections which
lead to an improved worst-case running time bound, we also devise the
following data reduction rules. These rules proved crucial for solving
larger instances of MCSP and may be of independent interest. The first
of these reduction rules identifies unique letters that are in~$S_1$
and~$S_2$ surrounded by other unique letters.
\begin{rrule}\label{rule:unique-border}
  If the instance contains a unique candidate match~$\{u,v\}$ and the
  letters to the right and left of~$u$ and~$v$ are also unique, then
  let~$L(u)$,~$R(u)$,~$L(v)$, and~$R(v)$ denote the uniquely defined
  candidate matches containing the left and right neighbor of~$u$
  or~$v$. Remove~$u$ and~$v$ from~$S_1$ and~$S_2$ and do the
  following.
  \begin{itemize}
  \item If~$L(u)= L(v)$ or~$R(u)=R(v)$ leave~$k$ unchanged.
  \item Else,  check whether removing~$u$
    and~$v$ from~$S_1$ and~$S_2$ made either~$L(u)$ and~$R(u)$
    parallel or~$L(v)$ and~$R(v)$ parallel. If it makes none of the
    two parallel, then decrease~$k$ by one, if it makes exactly one pair
    parallel, decrease~$k$ by two, otherwise decrease~$k$ by three.
   \end{itemize}
\end{rrule}
\begin{proof}[of correctness]
  In the first case,~$\{u,v\}$ is parallel to either~$L(u)$ or~$R(u)$
  and thus the rule is simply a special case of the parallel rule. In
  the other cases,~$\{u,v\}$ is parallel to none
  of~$L(u),R(u),L(v),R(v)$. Hence,~$u$ and~$v$ will be in a block of
  size one in any CSP. In case the removal of~$\{u,v\}$ makes no other
  edges parallel, the minimum size of a CSP in the reduced instance
  thus is one less. Hence, the parameter decrement is correct in this
  case. If the removal of~$u$ and~$v$ makes only~$L(u)$ and~$R(u)$
  parallel, then the minimum size of a CSP after removing~$u$ is
  decreased by exactly two: Consider any CSP of the original instance,
  ``merging'' the blocks containing the left and the right neighbor
  of~$u$ and removing the blocks containing~$u$ and~$v$ gives a CSP
  for the reduced instance with size decreased by two. Similarly,
  re-adding~$\{u,v\}$ to any CSP of the reduced instance increases the
  size by exactly two.  By symmetry, the same holds for the case that
  the removal of~$u$ and~$v$ makes only~$L(v)$ and~$R(v)$ parallel.

  Finally, if the removal makes $L(u)$ and~$R(u)$ parallel and $L(v)$
  and~$R(v)$ parallel, then the size of the minimum CSP decreases by
  exactly three which follows from the above arguments with the
  additional observation that the two block merges are indeed
  ``different''. \qed
\end{proof}

The next two rules ``split'' letters into two ``subletters''. The
first rule looks for letters that appear once in one sequence and
twice in the other.
\begin{rrule}\label{rule:star}
  If there is a marker~$v$ such that there is exactly one candidate
  match~$\{u,v\}$ containing~$v$, the marker~$u$ has at least one further
  candidate match~$\{u,w\}$, and any CSP which
  contains~$\{u,v\}$ has~$u$ and~$v$ in blocks of size one, then
  change the letter of~$v$ to some previously unused letter~$z$.
\end{rrule}
\begin{proof}[of correctness]
  Any CSP~$P$ of size~$k$ containing the candidate match~$\{u,v\}$ can
  be transformed into a CSP of size at most~$k$ containing the candidate
  match~$\{u,w\}$: Since~$u$ and~$v$ are in~$P$ in blocks of size one,
  replacing $\{u,v\}$ by~$\{u,w\}$ does not decrease the number of
  adjacencies in the blocks of the CSP. Furthermore, this exchange is
  possible, since~$\{u,w\}$ is the only candidate match
  containing~$w$. Hence, there is an optimal CSP in which~$v$ is not
  contained in any candidate match. It is thus safe to assign~$v$ some
  new unused letter. \qed
\end{proof}

The next rule follows the same idea, only with letters that appear
twice.
\begin{rrule}\label{rule:K22}
  If there is a set of four markers~$u$,~$v$,~$w$, and~$z$ such
  that~$\{u,w\}$, $\{u,z\}$, $\{v,w\}$,~$\{v,z\}$ are the only four
  candidate matches containing at least one of these markers, and any CSP which contains~$\{u,w\}$ and~$\{v,z\}$
  has~$u$,~$v$,~$w$, and~$z$ in blocks of size one, then change the
  letter of~$u$ and~$z$ to some previously unused letter~$x$. 
\end{rrule}
\begin{proof}[of correctness]
  The proof is similar to the proof of~\cref{rule:star}. Since the
  blocks containing~$u$,~$v$,~$w$, and~$z$ have size one, changing the
  candidate matches does not decrease the number of adjacencies in the
  blocks. Hence, replacing~$\{u,w\}$ and~$\{v,z\}$ by~$\{u,z\}$
  and~$\{v,w\}$ gives a CSP of the same size. 
\qed
\end{proof}
Note that checking whether there is any CSP including some match~$\{u,v\}$
that has~$u$ and~$v$ in blocks of size at least two can be done by simply checking
whether~$\{u,v\}$ is parallel to a candidate match of its right or
left neighbor.

\section{Implementation \& Experiments}
\label{sec:experiments}
We implemented the described algorithm to assess its performance on
genomic and on synthetic instances. We furthermore added three
additional data reduction rules and demonstrate their effect on the
genomic instances. Although our algorithm and experiments concern
unsigned strings, they can be seen as a first step; the results being
more than encouraging, we will adapt, in the near future, our
algorithm to the signed (and unbalanced) case. We ran all our
experiments on an Intel(R) Core(TM) i5 M 450 CPU 2.40GHz machine with
2GB memory under the Ubuntu 12.04 operating system. The program is
implemented in Java and runs under Java~1.6. The source code is
available from~\url{http://fpt.akt.tu-berlin.de/mcsp/}.  The search
tree is implemented as described in Sections~\ref{sec:search-tree}
and~\ref{sec:param-k'}.  In addition to the data reduction rules
described in~Section~\ref{sec:data-red}, we
apply~\cref{rule:reduce-unique}. All data reduction rules are applied
in the beginning and also in each search tree node.

\paragraph{Genomic Data.} We performed experiments with genomic
data from several bacteria. The data was obtained as follows. The raw
data consists of a file containing transcripts and proteins of the
species and positional information of the corresponding genes. This
data was downloaded from the EnsemblBacteria database~\cite{KSL+12}
and then filtered as described by~\citet{SZJ10} to obtain input data
for MSOAR 2.0. Then, the MSOAR 2.0 pipeline was invoked, and the MCSP
instances are output right before they are solved approximately by the vertex cover 2-approximation algorithm. These instances contain signed
genes. Since the presented correctness proof only solves the unsigned
MCSP problem, we removed all genes from the negative
strand. Afterwards, we removed all non matched genes. Finally, we
perform the following modification: the data from MSOAR actually can
allow arbitrary candidate matches between markers in~$S_1$
and~$S_2$. However in MCSP the candidate matches are ``transitive'',
that is, if~$\{u,v\}$,~$\{v,w\}$, and~$\{w,x\}$ are candidate matches of
an MCSP instance, then~$\{u,x\}$ is also a candidate match. We achieve
this property for the input data by adding the candidate
match~$\{u,x\}$, that is, every connected component of the
``marker-match'' graph is assigned one letter not used elsewhere.

The species under consideration are 
\textit{Borrelia burgdorferi}, \textit{Treponema pallidum}, \textit{Escherichia coli},
\textit{Bacillus subtilis}, and \textit{Bacillus thuringiensis}. Our
results are shown in~\cref{tab:genomic-results}; the main findings
are as follows. We can solve instances with 
hundreds of genes if the average number~$d^*$ of
occurrences for each letter and the number~$k'$ 
of blocks without unique letters is small. Moreover, the
parameter~$k'$ is in these instances much smaller than the
parameter~$k$. Finally, the data reduction rules are very effective in
decreasing the instance size and also decrease the overall number of
candidate matches somewhat.
\begin{table}[t]\centering \footnotesize{
  \caption{Running time, instance properties and effect of data reduction on genomic data. Herein, $n_1$ is the number of markers in the first genome,~$n_2$ the number of markers in the second genome,~$k$ is the CSP size,~$k'$ the number of blocks without fixed markers, $d^*$ the average number of candidate matches for each marker,~$n_1'$ and~$n'_2$ denote the respective number of markers after data reduction,~$\delta$ is the number of removed candidate matches during data reduction, and~$t$ is the running time in seconds.}  
\label{tab:genomic-results}
\begin{tabular*}{1\linewidth}{@{\extracolsep{\fill} } ll ccc cccc ccc}
\hline 
  Species~1 & Species~2 & $n_1$ & $n_2$ & $k$ &  $k'$ & $d$ & $d^*$ & $n'_1$ & $n'_2$ & $\delta$ & $t$  \\ \hline 
  \textit{B.~burg.}  & \textit{T.~pall.}  & 91&93& 68 & 0 & 3  & 1.02 &13&15 &4  & 0.06 \\ 
  \textit{B.~burg.}  & \textit{E.~coli}  & 66&72&59 & 0 &  6  & 1.09 &22&28&12  & 0.22 \\ 
  \textit{B.~burg.}  & \textit{B.~sub.}  & 83&91& 63& 3 & 6 & 1.16 &31&39&11 & 0.15 \\ 
  \textit{B.~burg.}  & \textit{B.~thur.} & 61&71& 51 & 3 & 5 & 1.19 &32&42&11 &0.09  \\
  \textit{T.~pall.}  & \textit{E~coli} & 89&93& 78 & 2 &  5 & 1.09 &22&26&7 & 0.35 \\ 
  \textit{T.~pall.}  & \textit{B.~sub.} & 136&144& 82 & 0  &  7 & 1.12 &23&31&11
 & 0.18 \\ 
  \textit{T.~pall.}  & \textit{B.~thur.} & 116&128& 76 & 0 &  6 & 1.16 &30&42&16  & 0.15 \\
  \textit{E.~coli}  & \textit{B.~sub.} & 264&287& 234 & 14 &  7 & 1.23 &128&151&54 &41.06  \\ 
  \textit{E.~coli}  & \textit{B.~thur.} & 249&282&221 &12 &  10 & 1.24 &129&162&59 & 18.64 \\
  \textit{B.~sub.}  & \textit{B.~thur.} & 673&693& 340 & 14  & 8 & 1.17 & 173&193&51  & 249.71\\ \hline
\end{tabular*} }
\end{table}

\paragraph{Synthetic Data.}
We also experimented with synthetic data to test how growth of~$k$
influences the running time. Each instance is generated randomly given five parameters: the string
length~$n$, the upper bound~$k$ on the number of blocks, the upper
bound~$d$ on the number of occurrences, the upper bound~$f$ on the
number of gene families (size of the alphabet), and finally the
number~$\delta$ of deleted markers (considered as \emph{noise} between
the blocks).  We randomly generate $k$ blocks using available markers
(that is, each block is a random string of markers so that the number of
occurrences is never more than $d$). The two input sequences are
generated by concatenating the blocks in different (random) orders,
interleaving with noisy parts of the required total size. 

We study the effect of varying parameters $n$, $k$ and $d$. To this
effect, we fix the number of deleted markers to $\delta=0.1n$ (we
observed that the behavior of the algorithm is uniform for $0\leq
\delta\leq 0.2n$).  Values of $\delta> 0.2n$ are harder, however, we
assume that deleting too many markers is of less relevance in genomic
applications.  The number~$f$ of gene families is fixed to
$3n/d$. This way we obtain an average number of occurrences which is
experimentally close to $d/2$.  The average occurrence of each letter
thus is roughly twice that of the genomic data; this was done to
obtain more difficult input.

In the experiments, we set~$n=1000$, and varied~$k$ from~$50$
to~$130$. One run was performed for~$d=6$ and one for~$d=8$. Our
results are shown in~\cref{tab:synth-d6}. For each set of parameter
values, we generated 50 instances. We make the following main
observations. First, increasing~$d$ makes the instances much
harder. Second, for~$d=6$, the combinatorial explosion sets in
at~$k\approx 120$, for~$d=8$ this happens already at~$k\approx
100$. Finally, the algorithm efficiently solves instances with~$n=1000$
and~$k\approx 120$ when the average occurrence of each letter is
roughly 3.5 (this is the average occurrence number in the experiments
for~$d=8$).
 
 \begin{table}[t]
   \centering
   \caption{Average running time in seconds for synthetic instances with~$d=6$ and~$d=8$,~$n=1000$ and varying~$k$; for each parameter triple, 50 instances were generated.}
   \label{tab:synth-d6}
   \begin{tabular*}{0.7\linewidth}{@{\extracolsep{\fill} }l l l l}
\hline 

 \multicolumn{2}{ c }{$d=6$} & \multicolumn{2}{ c }{$d=8$} \\

$k$ & running time & $k$ & running time  \\ \hline
50 &   0.06 &  50& 0.07 \\ 
60 &   0.06 &  60& 0.06 \\
70 &   0.07 &  70& 0.08\\
80 &   0.09 &  80& 0.09\\
90 &   0.10 &  90& 0.12\\
100 &   0.12 &  100& 0.16 \\
110 &   0.13 &  110& 0.26\\
120 &   0.18 &  120& 1.62\\
130 &   0.21 &  130& 30.42\\ \hline
   \end{tabular*}
 \end{table}

\section{Conclusion}

We have presented an efficient fixed-parameter algorithm for the \textsc{Minimum Common String Partition} problem
with parameters $k$ and $d$. Our algorithm even allows for unbalanced strings, since it can delete 
superfluous markers between consecutive blocks of the string partition. Looking towards practical 
applications, it would be interesting to consider signed instances, that is, blocks can be read either from
left to right or from right to left with opposite signs. 
We conjecture that our algorithm can be extended to solve the signed variant of MCSP. 
Another generalization of MCSP is as follows. Pairs of markers which form candidate matches are given
in input, rather than being defined from classes of letters. 
From a graph theory point of view, the 
bipartite graph of candidate matches may contain arbitrary connected components, 
not only complete ones.
It would be of interest to provide  efficient algorithms for this extension of MCSP.

\makeatletter
\renewcommand\bibsection%
{
  \section*{\refname
    \@mkboth{\MakeUppercase{\refname}}{\MakeUppercase{\refname}}}
}
\makeatother

{
\bibliographystyle{abbrvnat}
\bibliography{string-partition}
}

\begin{center}

\end{center}
\end{document}